\makeatletter
\documentclass[letterpaper, 10 pt, journal]{IEEEtran}  
\usepackage{amsmath,amsfonts,mathtools,amsthm}
\usepackage{dsfont, cases, tikz, subcaption}
\usepackage{hyperref}
\usepackage[capitalize,noabbrev,nameinlink]{cleveref}
\usepackage{algorithmic}
\usepackage[linesnumbered,ruled,vlined]{algorithm2e} 
\usepackage{url}
\usepackage{booktabs,multirow}
\usepackage[numbers,sort&compress]{natbib}
\crefformat{equation}{(#2#1#3)}
\Crefname{algocfline}{Algorithm}{Algorithms}
\Crefname{algocf}{line}{lines}
\Crefname{assumption}{Assumption}{Assumptions}
\crefrangeformat{assumption}{Assumptions~#3#1#4-#5#2#6}

\IEEEoverridecommandlockouts                              

\title{\LARGE \bf
A Convex Formulation of Game-theoretic Hierarchical Routing 
}

\author{Dong Ho Lee$^1$, Kaitlyn Donnel$^1$, Max Z. Li$^2$, and David Fridovich-Keil$^1$
\thanks{Dong Ho Lee, Kaityln Donnel and David Fridovich-Keil are with the Department of Aerospace Engineering and Engineering Mechanics, University of Texas at Austin, Austin, TX 78712, USA (e-mail: {\tt leedh0124@utexas.edu; kaitlyndonnel@utexas.edu; dfk@utexas.edu}).}
\thanks{Max Z. Li is with the Department of Aerospace Engineering, the Department
of Civil and Environmental Engineering, and the Department of Industrial and
Operations Engineering, University of Michigan, Ann Arbor, MI 48109 USA
(e-mail: {\tt maxzli@umich.edu}).}
\thanks{This work was supported by a National Science Foundation CAREER award under Grant No. 2336840. (\it {Corresponding author: Dong Ho Lee.})}%
}
\usepackage[acronyms, shortcuts]{glossaries}

\newcommand{\BibTeX}{\rm B\kern-.05em{\sc i\kern-.025em b}\kern-.08em\TeX}

\newtheorem{remark}{Remark}
\newtheorem{definition}{Definition}[section]
\newtheorem{proposition}{Proposition}

\newcommand{\ie}{i.e.}
\newcommand{\st}{s.t.}

\newcommand{\argmin}{\mathop{\rm argmin}}

\newcommand{\R}{\mathbb{R}}
\newcommand{\X}{\mathcal{Z}}
\newcommand{\Lagrangian}{\mathcal{L}}

\newcommand{\equality}{g}

\newcommand{\inequality}{h}

\newcommand{\jth}{$j^\mathrm{th}$~}
\newcommand{\kth}{$k^\mathrm{th}$~}
\newcommand{\kpth}{$(k+1)^\mathrm{th}$~}

\newcommand{\binvarz}{z}
\newcommand{\waypoints}{K}
\newcommand{\horizon}{T}
\newcommand{\nodes}{N}
\newcommand{\players}{M}
\newcommand{\bvar}{\mathbf{z}}

\newcommand{\xstate}{x}
\newcommand{\control}{u}

\newcommand{\xdim}{n}

\newcommand{\xposition}{p}
\newcommand{\velocity}{v}
\newcommand{\acceleration}{a}
\newcommand{\cost}[1]{J^{#1}}

\glsdisablehyper
\newacronym[longplural=Nash equilibria,plural=NEs]{ne}{NE}{Nash equilibrium}
\newacronym[longplural=open-loop Nash equilibria,plural=OLNE]{olne}{OLNE}{open-loop Nash equilibrium}
\newacronym{lq}{LQ}{linear quadratic}
\newacronym{minlp}{MINLP}{mixed integer nonlinear problem}
\newacronym{miqp}{MIQP}{mixed integer quadratic problem}
\newacronym{lp}{LP}{linear programming}
\newacronym{qp}{QP}{quadratic programming}

\newacronym{micp}{MiCP}{mixed complementarity problem}
\newacronym{kkt}{KKT}{Karush-Kuhn-Tucker}
\newacronym{mpec}{MPEC}{mathematical program with equilibrium constraints}
\newacronym{mpcc}{MPCC}{mathematical program with complementarity constraints}
\newacronym{gnep}{GNEP}{generalized Nash equilibrium problem}
\newacronym[longplural=generalized Nash equilibria,plural=GNEs]{gne}{GNE}{generalized Nash equilibrium}
\newacronym{nlp}{NLP}{nonlinear programming}
\newacronym{mfcq}{MFCQ}{Mangasarian-Fromowitz constraint qualification}
\newacronym{licq}{LICQ}{linear independence constraint qualification}
\newacronym{cq}{CQ}{constraint qualification}

\begin{document}

\maketitle
\thispagestyle{empty}
\pagestyle{empty}

\begin{abstract}

Hierarchical decision-making is a natural paradigm for coordinating multi-agent systems in complex environments such as air traffic management. In this paper, we present a bilevel framework for game-theoretic hierarchical routing, where a high-level router assigns discrete routes to multiple vehicles who seek to optimize potentially noncooperative objectives that depend upon the assigned routes. To address computational challenges, we propose a reformulation that preserves the convexity of each agent's feasible set. This convex reformulation enables a solution to be identified efficiently via a customized branch-and-bound algorithm. Our approach ensures global optimality while capturing strategic interactions between agents at the lower level. We demonstrate the solution concept of our framework in two-vehicle and three-vehicle routing scenarios.
\end{abstract}

\section{Introduction} \label{sec:intro}

Coordinating multiple vehicles in a controlled fashion is a crucial challenge for decision making in applications such as autonomous transportation systems and air traffic management.  
Such problems are inherently hierarchical:
routing decisions assign paths to multiple vehicles, and individual vehicles determine optimal trajectories in response to their assigned routes.
Current approaches treat these problems in isolation, i.e., higher-level routing decisions operate on discrete graphs which abstract away the continuous lower-level trajectory optimization process.
In this work, we present a technique for solving these coupled problems \emph{simultaneously}, and 
showcase its performance in an example inspired by air traffic management.
Ultimately, our results suggest that integrating the discrete (multi-)vehicle routing problem with low-level trajectory design can reduce vehicles' control effort and yield more efficient system-wide performance.

A natural approach to capturing this hierarchical structure is through bilevel optimization, where the upper level handles discrete decision making in the form of route assignments, and the lower level encodes continuous optimization of individual vehicles' trajectories.
Information flows between these levels bidirectionally.
The upper-level routing problem accounts for the physical state and operational requirements of individual vehicles, while, at the lower-level, the vehicles generate trajectories that respect the prescribed routing plan.

On their own, routing problems are well-studied in the context of (mixed) integer programs \cite{zeng2014solving,kleinert2021survey}.
Similarly, the lower-level trajectory optimization is well-explored in both single- and multi-agent contexts; in this work, we consider the more general multi-agent, potentially noncooperative, variants of these problems.

Concretely, we make the following contributions.
\emph{(i)} We formulate hierarchical vehicle routing and trajectory design problems as a mixed-integer bilevel program, where the upper level models a discrete routing problem and the lower level encodes a (potentially noncooperative) trajectory design game played among $\players$ vehicles. \emph{(ii)} We propose a reformulation that integrates the two levels, and solve the resulting problem via a customized branch-and-bound algorithm.  
Experimental results showcase the proposed approach in a scenario inspired by air traffic management.

\section{Related Work} \label{sec:related}


Current optimization models for air traffic management typically focus on strategic decisions, such as optimizing flight delays \cite{richetta1994dynamic,vranas1994multi} or airspace sector route assignments \cite{bertsimas1998air}. Recent work examines hierarchical congestion pricing and route planning via a bilevel optimization approach, but ignores vehicle dynamics entirely \cite{wu2025managing}. On the other hand, aircraft and air vehicle trajectory optimization studies (e.g., \cite{taye2024safe}) typically do not jointly optimize higher-level strategic routing or traffic flow management decisions.
This separation between strategic decision-making and trajectory optimization can lead to inefficiencies, as strategic plans may not fully account for vehicle dynamics and interactions at the trajectory level.
This work contributes an approach which begins to close this gap.

At the trajectory level, multi-agent interactions play a critical role in ensuring efficient airspace management. 
Game-theoretic approaches \cite{zhu2014game,williams2023distributed, 9197129,di2019newton,forrestsiopt,10021943,10160799} have been widely explored, and efficient methods exist which can compute approximate and/or local Nash equilibrium strategies in multi-agent, noncooperative settings.
However, these approaches generally assume a fixed environment \emph{without} a high-level planner that can influence individual agents' decisions. 

Hierarchical game-theoretic frameworks have also been proposed for trajectory optimization, primarily in the context of human-autonomous vehicle interactions \cite{fisac2019hierarchical,JI2023104109}. 
These works typically model motion planning as a dynamic Stackelberg game, where an autonomous vehicle (follower) anticipates and reacts to a human driver (leader). 

Existing works on bilevel optimization for routing problems \cite{marinakis2007new,fan2011bi,bianco2009bilevel} consider hierarchical decision-making between multiple stakeholders.
However, they often assume cooperative settings and do not model interactions between individual entities at the lower level. 
Moreover, their solution methods often rely on metaheuristic algorithms such as genetic algorithms, or involve linearizing nonconvex constraints to obtain tractable relaxed formulations. 

Unlike previous approaches, our work introduces a distinct hierarchical structure that explicitly integrates high-level strategic routing with (potentially noncooperative) low-level trajectory optimization.  
Our proposed mixed-integer bilevel framework enables joint optimization of both levels of decision-making for more efficient management of multiple vehicles. This has wide applicability for future air traffic management automation systems (cf. \cref{sec:conclusion}).

\section{Problem Formulation} \label{sec:problem}

In this section, we formulate a hierarchical routing game problem. This formulation can easily accommodate additional application-specific constraints, cf. \cref{sec:method}. 

\subsection{Preliminaries}
We represent the environment as a graph with $\nodes$ nodes, each corresponding to a point of interest to be visited (e.g., airspace navigational aids), and $\players$ vehicles, each generating their own continuous state/action trajectories.
The nodes are defined in a 2D Cartesian plane\footnote{This formulation readily extends to 3D environments.} as $\mathcal{P} := \{\hat{\mathbf{p}}_i\}_{i=1}^\nodes = \{(\hat x_i, \hat y_i)\}_{i=1}^\nodes$.
Additionally, each vehicle's trajectory starts at their own \emph{start} node $S^j$, whose position is $\hat{\mathbf{p}}^j_S = (\hat x^j_S, \hat y^j_S)$, and ends at \emph{terminal} node $T^j$, which is located at $\hat{\mathbf{p}}^j_T = (\hat x^j_T, \hat y^j_T)$. 
A \emph{waypoint} is any node included in the route;
a route has $\waypoints$ waypoints, and cannot have more than $\nodes$ waypoints, \ie, $\waypoints \leq \nodes$. Between each pair of consecutive waypoints, $\horizon$ intermediate \emph{subwaypoints} define a segment of the given vehicle's trajectory. 

\subsection{Proposed Bilevel Framework}
The proposed bilevel program is of the form:
%
\begin{subequations} 
\label{eqn:III-prelim-bilevel} \begin{align} 
&\min_{\substack{{\bvar_u \in \X_u}\\{\bvar_l \in \X_l}}} ~ \cost{}_u(\bvar_u, \bvar_l) \label{eqn:III-prelim-top-level-obj}\\ 
&\mathrm{\st} ~ \forall j \in [\players] \ \Bigg\{ \bvar_l^j \in \argmin_{\tilde\bvar_l^j\in \X_l^j(\bvar_u, \tilde\bvar_l^{\neg j})} \cost{j}_l(\bvar_u, \tilde\bvar_l)\,, \quad \label{eqn:III-prelim-gnep}
\end{align} 
\end{subequations}
where $\bvar_u \in \R^{\xdim_u}$ and $\bvar_l \in \R^{\xdim_l}$ represent the upper-level (discrete) and the lower-level (continuous) decision variables, respectively. 
Similarly, $  \cost{}_u(\cdot):\R^{\xdim_u} \times \R^{\xdim_l} \to \R $ and $ \cost{}_l(\cdot):\R^{\xdim_u} \times \R^{\xdim_l} \to \R $ denote the objective functions at the upper and lower level, respectively. 
The upper-level variables $\mathbf{\binvarz}_u$ are binary, i.e., $\X_u \coloneqq \{\underline \binvarz_u^k, \bar \binvarz_u^k\}^{n_u},$ where $\underline \binvarz_u^k=0$ and $\bar \binvarz_u^k=1$.
The superscript on $\bvar^j_l \in \R^{\xdim^j_l}$ in \cref{eqn:III-prelim-gnep} refers to the $j^\mathrm{th}$ vehicle's variables. 
We use $\bvar^{\neg j}_l \in \R^{\xdim_l - \xdim^j_l}$ to denote the variables of all vehicles except $j$.
$[\players]$ in \cref{eqn:III-prelim-gnep} denotes the set $\{1, \dots, M\}$.
In subsequent sections, we characterize the structure of the constraints and cost functions in \cref{eqn:III-prelim-bilevel}.

\subsection{Upper-level Routing Problem}
The upper level routing problem in \cref{eqn:III-prelim-bilevel} models a \emph{router} that determines the sequence of waypoints for all vehicles' routes, given fixed \emph{initial} and \emph{terminal} nodes for each vehicle. 
We formulate the upper level as a discrete problem that consists of binary decision variables $\mathbf{\binvarz}_u \coloneqq [\mathbf{\binvarz}^1, \dots, \mathbf\binvarz^{\players}]$ where $ \mathbf\binvarz^j \coloneqq [\binvarz_{S^j,1}^j, \dots, \binvarz_{S^j,K}^j, \binvarz_{1,1}^j, \dots \binvarz_{i,k}^j, \dots \binvarz_{N,K}^j, \dots \binvarz_{T^j,K}^j] \in \mathbb{R}^{\nodes \waypoints}, ~\forall j \in [\players]$.  
The binary variable $\binvarz_{i,k}^j$ is defined such that $\binvarz_{i,k}^j = 1$ if node $i$ is the $k^\mathrm{th}$ waypoint for vehicle $j$, and $\binvarz_{i,k}^j = 0$ otherwise.
$\X_u$ in \cref{eqn:III-prelim-top-level-obj} denotes the feasible set for $\mathbf{\binvarz}_u$ and is defined in terms of constraints that must be satisfied for a feasible route. 
By definition, the following constraints naturally encode the beginning and the end of a trajectory: 
%
\begin{equation}
\label{eqn:III-top-level-node-equality1}
\binvarz_{S^j,1}^j = 1, \binvarz_{T^j,K}^j = 1, ~\forall j \in [\players].
\end{equation}
No vehicle can be assigned two waypoints at the same time, and at most one vehicle can be assigned any waypoint at the same time. 
These constraints are encoded by: 
\begin{align}
\label{eqn:III-top-level-node-equality2}
&\sum_{i=1}^{\nodes} \binvarz_{i,k}^j = 1, ~\forall k \in [\waypoints], j \in [\players], \\
\label{eqn:III-top-level-inequality1}
&\sum_{j=1}^{\players} \binvarz_{i,k}^j \leq 1, ~\forall i \in [\nodes]\cup \{S^j,T^j\}, k \in [\waypoints].
\end{align}
Furthermore, a node may not be visited more than once by any vehicle, i.e.,
\begin{equation}
\label{eqn:III-top-level-node-inequality3}
\sum_{k=1}^{\waypoints} \binvarz_{i,k}^j \leq 1, ~\forall i \in [\nodes] \cup \{S^j,T^j\}, j \in [\players].
\end{equation}
\Cref{eqn:III-top-level-node-equality1,eqn:III-top-level-node-equality2,eqn:III-top-level-node-inequality3,eqn:III-top-level-inequality1} together define $\X_u$ in \cref{eqn:III-prelim-top-level-obj}. 

\subsection{Lower-level Trajectory Game}
The lower-level trajectory planning problem in \cref{eqn:III-prelim-gnep} is mathematically formulated as a multi-agent trajectory game, whose solution is a \ac{ne} \cite{basar1998}.
The lower-level decision variables $\bvar_l \coloneqq [\bvar^1_l, \dots, \bvar_l^M]$ are continuous, and 
for each vehicle $j$, $\bvar^j_l \coloneqq (\mathbf\xstate^j, \mathbf\control^j) = (\xstate^j_{1:\waypoints,1:\horizon}, \control^j_{1:\waypoints,1:\horizon})$ where $\xstate^j_{k,t} \in \R^{n_x}$ and $\control^j_{k,t} \in \R^{n_u}$ represent the state and the control input variables of the \jth vehicle.
We consider a variant of \ac{lq} games \cite{dirkse1995path} for this level, \ie, \cref{eqn:III-prelim-gnep} is in the form of:
\begin{subequations} 
\label{eqn:III-prelim-lqgame} \begin{align} 
&\min_{\mathbf\xstate^j, \mathbf\control^j} ~ \overbrace{\frac{1}{2} \sum_{\substack{{k=1}\\{t=1}}} \Bigl(Q(\xstate^j_{k,t},\binvarz_{1:N,k}^j) + \|\control_{k,t}^j\|_2^2 \Bigr) + I(\mathbf\xstate^j, \mathbf\xstate^{\neg j})}^{\cost{j}_l(\mathbf\binvarz_u, \mathbf\xstate, \mathbf\control^j)} \label{eqn:III-prelim-lqgame-obj}\\  
&\mathrm{\st \quad} \xstate_{k,{t+1}}^j = A\xstate_{k,t}^j + B\control^j_{k,t}, ~\forall k<\waypoints, ~\forall t<\horizon, \label{eqn:III-prelim-lqgame-constraint1} \\
&\phantom{\mathrm{\st \quad}} \xstate_{{k+1},1}^j = A\xstate_{k,T}^j + B\control^j_{k,T}, ~\forall k<\waypoints, \label{eqn:III-prelim-lqgame-constraint2} \\
&\phantom{\mathrm{\st \quad}} \underline{\xstate} \leq \xstate^j_{k,t} \leq \overline{\xstate}, ~\forall k\in[\waypoints], ~\forall t\in[\horizon], \label{eqn:III-prelim-lqgame-state-bounds} \\
&\phantom{\mathrm{\st \quad}} \underline{\control} \leq \control^j_{k,t} \leq \overline{\control}, ~\forall k\in[\waypoints], ~\forall t\in[\horizon]. \label{eqn:III-prelim-lqgame-control-bounds}
\end{align} 
\end{subequations}

\cref{eqn:III-prelim-lqgame-constraint1,eqn:III-prelim-lqgame-constraint2,eqn:III-prelim-lqgame-control-bounds,eqn:III-prelim-lqgame-state-bounds} together characterize $\X_l$ in \cref{eqn:III-prelim-bilevel}.
The subscripts $(k,t)$ in \cref{eqn:III-prelim-lqgame} correspond to time step $t$, after passing the \kth waypoint.
Specifically, \cref{eqn:III-prelim-lqgame-constraint1} describes state transition within the trajectory segment $k$, \ie, from the \kth to the \kpth waypoint. 
Similarly, \cref{eqn:III-prelim-lqgame-constraint2} defines the changes from the final time step $\horizon$ of segment $k$ to the initial time step of segment $k+1$. 
\cref{eqn:III-prelim-lqgame-state-bounds,eqn:III-prelim-lqgame-control-bounds} encode the bounds on the state and control variables for the \jth vehicle.
We discuss the cost structure \cref{eqn:III-prelim-lqgame-obj} of the lower-level problem in the next section. 

\subsection{Structure of Cost Functions}
We now characterize the structure of objective functions at the upper \cref{eqn:III-prelim-top-level-obj} and lower \cref{eqn:III-prelim-gnep} levels. 
The upper-level routing problem minimizes the total control effort, as a proxy for fuel consumption, i.e., $\cost{}_u(\bvar_u, \bvar_l) \coloneqq \sum_{j=1}^\players \sum_{k=1}^\waypoints \sum_{t=1}^\horizon \|\control^j_{k,t}\|^2_2$.
Note that the variable $\control^j_{k, t}$ \cref{eqn:III-prelim-lqgame} is determined at the lower level \emph{in response to} the routing choice $\mathbf\binvarz_u$ \cref{eqn:III-prelim-bilevel} at the upper level.


The cost function at the lower level, $\cost{j}_l(\mathbf\binvarz_u, \mathbf\xstate, \mathbf\control^j)$, in \cref{eqn:III-prelim-lqgame-obj} is comprised of two parts. 
The first term captures individual vehicle costs such as state penalities and control effort, while the second term accounts for interaction effects among vehicles. 
The function $Q(\xstate^j_{k,t},\binvarz_{1:N,k}^j)$ is convex and quadratic with respect to both $\xstate^j_{k,t}$ and $\binvarz_{1:N,k}^j$. 
In \cref{sec:experiment}, we describe how $Q(\xstate^j_{k,t},\binvarz_{1:N,k}^j)$ can penalize trajectory deviations from the upper-level routing plan. 
The inter-agent interaction term $I(\mathbf\xstate^j, \mathbf\xstate^{\neg j})$ is also a convex quadratic function, and encodes the influence of other vehicles' trajectories on vehicle $j$'s cost. 
For example, in \cref{sec:experiment} we use this interaction term to model aircraft formation flight.

The lower-level game consists of linear constraints and quadratic objective functions for all vehicles, \ie, the lower-level game is a convex game parameterized by $\mathbf\binvarz_u$.
Thus, the \ac{kkt} conditions for each agent's problem are both necessary and sufficient for optimality \cite{basar1998}.


\section{Methodology} \label{sec:method}

In this section, we present a \ac{kkt} reformulation \cite{dempe2013bilevel,dempe2015bilevel} of the hierarchical routing game in \cref{eqn:III-prelim-bilevel}. 
We introduce auxillary variables to transform the problem into a \ac{miqp}, and develop a customized branch-and-bound algorithm for solving the resulting problem.

\subsection{\ac{kkt} Reformulation}
For the ease of discussion, we express the lower-level \ac{lq} game \cref{eqn:III-prelim-lqgame} in the following form:
\begin{subequations} 
\label{eqn:IV-method-lqgame} \begin{align} 
\textit{(Vehicle j's problem)}\quad &\min_{\mathbf\xstate^j, \mathbf\control^j} \cost{j}_l(\mathbf\binvarz_u, \mathbf\xstate, \mathbf\control^j) \qquad\qquad \label{eqn:IV-method-lqgame-obj}\\  
&\mathrm{\st \quad} \equality(\mathbf\xstate^j,\mathbf\control^j) = 0, \label{eqn:IV-method-equality-constraint} \\
&\phantom{\mathrm{\st \quad}} \inequality(\mathbf\xstate^j,\mathbf\control^j) \geq 0 \label{eqn:IV-method-inequality-constraint}, 
\end{align} 
\end{subequations}
where $\mathbf\xstate \coloneqq [\mathbf\xstate^j, \mathbf\xstate^{\neg j}]$.
Equation \cref{eqn:IV-method-equality-constraint} encodes the equality constraints from \cref{eqn:III-prelim-lqgame-constraint1,eqn:III-prelim-lqgame-constraint2}. 
These govern the dynamics of the \jth vehicle's trajectory.
Similarly, \cref{eqn:IV-method-inequality-constraint} refers to the bounds in \cref{eqn:III-prelim-lqgame-control-bounds,eqn:III-prelim-lqgame-state-bounds}, which encode the state and control bounds. 

The original bilevel problem \cref{eqn:III-prelim-bilevel} and its single-level reformulation using \ac{kkt} conditions are equivalent under the assumption that the lower-level problem is convex for each vehicle $j$ \emph{and} that Slater's condition holds for all feasible upper-level variables, $\mathbf\binvarz_u \in \X_u$. 

\begin{definition}[Slater's condition for the lower level]
\label{def:Slater}
For a given feasible $\mathbf\binvarz_u \in \X_u$ of \cref{eqn:III-prelim-bilevel}, Slater's constraint qualification holds for the lower-level problem if there exists $\tilde{\mathbf\xstate}^j,\tilde{\mathbf\control}^j$ \st ~$\equality(\tilde{\mathbf\xstate}^j,\tilde{\mathbf\control}^j) = 0$ and $\inequality(\tilde{\mathbf\xstate}^j,\tilde{\mathbf\control}^j) > 0$. 
\end{definition}

Under these assumptions, we can rewrite \cref{eqn:IV-method-lqgame} using its \ac{kkt} conditions. 
We obtain the following single-level reformulation: 
\begin{subequations}
\label{eqn:IV-method-kkt-formulation}
\begin{align}
\min_{\mathbf\binvarz_u,\mathbf\binvarz_l,\lambda,\mu} \quad & \cost{}_u(\mathbf\binvarz_u, \mathbf\binvarz_l) \\
\mathrm{\st \qquad \ ~} & \mathbf\binvarz_u \in \X_u, \\
& \nabla_{\mathbf\xstate^j,\mathbf\control^j} \Lagrangian(\mathbf\binvarz_u,\mathbf\xstate,\mathbf\control^j,\lambda^j,\mu^j) = 0, \label{eqn:IV-method-kkt-stationary} \\
&  0 \leq \inequality(\mathbf\xstate^j,\mathbf\control^j) \ \bot \ \lambda_j \geq 0, \label{eqn:IV-method-complementarity}\\
&  \equality(\mathbf\xstate^j,\mathbf\control^j) = 0, ~\quad \forall j \in [\players], \label{eqn:IV-method-equality}  
\end{align}
\end{subequations}
where the \jth agent's \emph{Lagrangian} function is defined as $\Lagrangian(\mathbf\binvarz_u,\mathbf\xstate,\mathbf\control^j,\lambda^j,\mu^j) \coloneqq \cost{j}_l(\mathbf\binvarz_u, \mathbf\xstate, \mathbf\control^j) - \lambda^{j\top} \inequality(\mathbf\xstate^j,\mathbf\control^j) - \mu^{j\top} \equality(\mathbf\xstate^j,\mathbf\control^j)$. 
Note that $\mu^j$ and $\lambda^j$ refer to the dual variables corresponding to the inequality and equality constraints of the \jth vehicle in \cref{eqn:IV-method-lqgame}, respectively.

We observe that the \ac{kkt} reformulation introduces complementarity constraints in \cref{eqn:IV-method-complementarity}, i.e., $\inequality(\mathbf\xstate^j,\mathbf\control^j) \ \bot \ \lambda_j \equiv \inequality(\mathbf\xstate^j,\mathbf\control^j)^\top \lambda_j = 0$, which involve bilinear terms that break the convexity of the feasible set in \cref{eqn:IV-method-kkt-formulation}. 
For general \ac{nlp} problems at the upper level, these complementarity constraints \cref{eqn:IV-method-kkt-formulation} can be handled using relaxation-based approaches \cite{scholtes2001convergence,schwartz2011mathematical}. 
However, since the upper level variables $\binvarz_u$ are binary, this formulation results in a nonconvex \ac{minlp}, which is intractable to solve.
A standard strategy in the \ac{minlp} literature is to obtain convex relaxations of the underlying nonconvex problems \cite{belotti2013mixed,kronqvist2019review}.
However, rather than developing tighter convex relaxations of an inherently nonconvex problem, we propose a more direct approach: reformulating the problem to preserve convexity from the outset.

\subsection{Introducing Auxillary Variables to Preserve Convexity}
The nonconvexity of \cref{eqn:IV-method-kkt-formulation} arises from the inequality constraints in \cref{eqn:IV-method-inequality-constraint} of the \ac{lq} game. 
To bypass this issue, we introduce auxillary variables at the upper level to \emph{reconstruct} the equilibrium trajectories from the lower-level problems. 
To this end, we define  auxillary variables  $\Breve{\bvar}_u \coloneqq (\mathbf{\Breve{\xstate}}^j, \mathbf{\Breve{\control}}^j) = (\Breve{\xstate}^j_{1:\waypoints,1:\horizon}, \Breve{\control}^j_{1:\waypoints,1:\horizon})$, and obtain the following problem:
\begin{subequations} 
\label{eqn:IV-method-bilevel-auxillary} 
\begin{alignat}{3} 
&\min_{\bvar_u, \Breve{\bvar}_u, \bvar_l} && \cost{}_u(\Breve{\bvar}_u, \bvar_l) + \alpha\sum_{j=1}^M \bigg\|\begin{bmatrix}
        \mathbf{\Breve{\xstate}}^j \\
        \mathbf{\Breve{\control}}^j 
    \end{bmatrix} - \begin{bmatrix}
        \mathbf{\xstate}^j \\
        \mathbf{\control}^j
    \end{bmatrix} \bigg\|_2^2 \label{eqn:IV-method-bilevel-auxillary-top-level-obj} \\
& \mathrm{\st \quad} && \equality(\mathbf{\Breve{\xstate}}^j,\mathbf{\Breve{\control}}^j) = 0, ~\inequality(\mathbf{\Breve{\xstate}}^j,\mathbf{\Breve{\control}}^j) \geq 0, && \forall j \in [\players], \\
& && \bvar_l^j \in \argmin_{\tilde\bvar_l^j\in \tilde{\X}_l^j(\bvar_u, \tilde\bvar_l^{\neg j})} \cost{j}_l(\bvar_u, \tilde\bvar_l), && \forall j \in [\players], \label{eqn:IV-method-bilevel-auxillary-gnep}
\end{alignat}
\end{subequations}

where $\cost{}_u(\Breve{\bvar}_u, \bvar_l) \coloneqq \sum_{j=1}^\players \sum_{k=1}^\waypoints \sum_{t=1}^\horizon \|\Breve{\control}^j_{k,t}\|^2_2$ and $\alpha>0$ is a weight parameter that penalizes large deviations from $(\mathbf{\xstate}^j, \mathbf{\control}^j)$. 
We can interpret $(\mathbf{\xstate}^j, \mathbf{\control}^j)$ as a \emph{reference} trajectory that arises from the noncooperative game played at the lower level (i.e., accounting for individual and inter-agent interaction costs, cf. \cref{sec:problem}). 
Note that $\tilde{\X}_l^j \coloneqq \{(\mathbf\xstate^j,\mathbf\control^j) \, | \, \equality(\mathbf\xstate^j,\mathbf\control^j) = 0\}$ in \cref{eqn:IV-method-bilevel-auxillary-gnep}.
Similarly,  $(\mathbf{\breve{\xstate}}^j, \mathbf{\Breve{\control}}^j)$ refers to the \emph{adjusted} trajectory (for the \jth vehicle) that must still satisfy the dynamics $\equality(\Breve{\mathbf\xstate}^j,\Breve{\mathbf\control}^j) = 0$ and state-control bounds $\inequality(\Breve{\mathbf\xstate}^j,\Breve{\mathbf\control}^j) \geq 0$. 

With the auxillary variables defined in \cref{eqn:IV-method-bilevel-auxillary}, we can write its \ac{kkt} reformulation as follows:
\begin{subequations} 
\label{eqn:IV-method-bilevel-auxillary-reform} 
\begin{alignat}{3} 
&\min_{\substack{\bvar_u, \Breve{\bvar}_u,\\\bvar_l, \mu}} && \cost{}_u(\Breve{\bvar}_u, \bvar_l) + \alpha\sum_{j=1}^M \bigg\|\begin{bmatrix}
        \mathbf{\Breve{\xstate}}^j \\
        \mathbf{\Breve{\control}}^j 
    \end{bmatrix} - \begin{bmatrix}
        \mathbf{\xstate}^j \\
        \mathbf{\control}^j
    \end{bmatrix} \bigg\|_2^2 \label{eqn:IV-method-bilevel-auxillary-top-level-obj-reform} \\
&\mathrm{\st \quad } && \equality(\mathbf{\Breve{\xstate}}^j,\mathbf{\Breve{\control}}^j) = 0, ~\inequality(\mathbf{\Breve{\xstate}}^j,\mathbf{\Breve{\control}}^j) \geq 0, && \forall j \in [\players], \label{eqn:IV-method-auxillary-reform-constraint1}\\
& && \nabla_{\mathbf\xstate^j,\mathbf\control^j} \tilde{\Lagrangian}(\mathbf\binvarz_u,\mathbf\xstate,\mathbf\control^j,\mu^j) = 0, && \forall j \in [\players],\label{eqn:IV-method-auxillary-reform-constraint2}\\
& && \equality(\mathbf\xstate^j,\mathbf\control^j) = 0,  && \forall j \in [\players], \label{eqn:IV-method-auxillary-reform-constraint3}
\end{alignat}
\end{subequations}
where $\tilde{\Lagrangian}(\mathbf\binvarz_u,\mathbf\xstate,\mathbf\control^j,\mu^j) = \cost{j}_l(\mathbf\binvarz_u, \mathbf\xstate, \mathbf\control^j) - \mu^{j\top} \equality(\mathbf\xstate^j,\mathbf\control^j)$.
Note that the \ac{kkt} conditions of the lower-level game in \cref{eqn:IV-method-bilevel-auxillary-gnep} are free of inequality constraints.
This, in turn, shows that the single-level reformulation \cref{eqn:IV-method-bilevel-auxillary-reform} does not have complementarity constraints as in \cref{eqn:IV-method-complementarity}.
As a consequence, the feasible set in \cref{eqn:IV-method-bilevel-auxillary-reform} remains convex (apart from the binary nature of the variables $\mathbf{\binvarz}_u$). 
This implies that, ultimately, problem \cref{eqn:IV-method-bilevel-auxillary-reform} is a \acf{miqp}.
With this transformation, any continuous relaxation of the upper-level binary variables $\mathbf{\binvarz}_u$ in \cref{eqn:IV-method-bilevel-auxillary-reform} yields a convex \acf{qp} problem which can be efficiently solved to obtain a lower bound on the optimal cost of problem \cref{eqn:IV-method-bilevel-auxillary-reform} in a branch-and-bound scheme.
We discuss the resulting algorithm leveraging this reformulation in the next section.

\subsection{Branch-and-bound algorithm}
We present our branch-and-bound method algorithm for solving problem \cref{eqn:IV-method-bilevel-auxillary-reform} in \cref{branch-and-bound} and discuss its convergence properties in \cref{proposition-convergence}. 

At any iteration $k$, \cref{branch-and-bound} solves a \emph{relaxed} version of \cref{eqn:IV-method-bilevel-auxillary-reform} in which the binary variables are treated as continuous and restricted to specific upper and lower bounds, i.e., $\underline{\mathbf{\binvarz}}_u^k$ and $\bar{\mathbf{\binvarz}}_u^k$ (cf. $\X_u$ in \cref{eqn:III-prelim-bilevel}).
If the relaxed problem results in a solution that is less optimal than the best solution obtained so far, i.e., $\tilde{\cost{k}_u} \geq \cost{*}_u$, we prune it, and thereby eliminate the need to expore all potential subproblems with tighter relaxations on the binary variables. 
If the solution contains binary values for $\mathbf\binvarz^k_u$, then the algorithm examines the resulting cost and, if superior to the best it has already found (i.e., $\tilde{\cost{k}_u} < \cost{*}_u$), records $\mathbf\binvarz^k_u$ as the new candidate optimum; otherwise, the feasible set is partitioned along a fractional-valued coordinate in $\mathbf\binvarz^k_u$ to create two further subproblems.

\begin{algorithm}
\label{branch-and-bound}
    \caption{Branch-and-Bound Algorithm}
    Initialize: Root node $N\coloneqq (\cost{}_u,\underline{\mathbf{\binvarz}}_u, \bar{\mathbf{\binvarz}}_u)$, queue $\mathcal{Q}$, $\bvar_u^* \gets \emptyset$, $\cost{*}_u \gets \infty$, $\cost{}_u \gets -\infty$, $\underline{\mathbf{\binvarz}}_u \gets \mathbf{0}$, $\bar{\mathbf{\binvarz}}_u \gets \mathbf{1}$, $Q \gets Q \cup \{N\}$, $k \gets 0$ \\
    \While{$\mathcal{Q}$ is not empty}{
        $N^k = (\cost{k}_u,\underline{\mathbf{\binvarz}}_u^k, \bar{\mathbf{\binvarz}}_u^k) \gets Dequeue(\mathcal{Q})$ \\
        \If{$\cost{k}_u \geq \cost{*}_u$}{
            \bf{continue}\tcp*{Prune (suboptimal)}
        }
        Solve the relaxed problem at $N^k$ to get $\bvar^k_u$ and $\tilde{\cost{k}_u}$ \label{alg:relaxed_problem} \\
        \If{$N^k$ is infeasible}{
            \bf{continue}\tcp*{Prune (infeasible)}
        }
        \Else{
            \If{$\tilde{\cost{k}_u} \geq \cost{*}_u$}{
                \bf{continue}\tcp*{Prune (suboptimal)}
            }
        }
        
        \If{$\bvar^k_u$ are binary}{
            \If{$\tilde{\cost{k}_u} < \cost{*}_u$}{
                $\bvar^*_u \gets \bvar^k_u, \cost{*}_u \gets \tilde{\cost{k}_u}$\tcp*{New best soln}
            }
            \bf{continue} 
        }
        \Else{ \tcp{Create subproblems}
            \If{$\tilde{\cost{k}_u} < \cost{*}_u$}{
                Select index $j$ such that $\bvar^k_u(j) \notin \mathbb{Z}$ \\ 
                $
                N_{\text{left}} \coloneqq (\tilde{\cost{k}_u}, \underline{\mathbf{\binvarz}}_u^k, \bar{\mathbf{\binvarz}}_u^k), \quad \bar{\mathbf{\binvarz}}_u^k(j) \gets \lfloor \bvar^k_u(j) \rfloor$
                
                $
                N_{\text{right}} \coloneqq (\tilde{\cost{k}_u}, \underline{\mathbf{\binvarz}}_u^k, \bar{\mathbf{\binvarz}}_u^k), \quad \bar{\mathbf{\binvarz}}_u^k(j) \gets \lceil \bvar^k_u(j) \rceil
                $
                
                $\mathcal{Q} \gets \mathcal{Q} \cup \{N_{\text{left}}, N_{\text{right}}\}$
            }
        }
        $k \gets k+1$
    }
    \Return $\bvar^*_u$, $\cost{*}_u$
\end{algorithm}


\begin{proposition}
\label{proposition-convergence}
    Under the following conditions, \cref{branch-and-bound} finds a global optimal solution to \cref{eqn:IV-method-bilevel-auxillary-reform}:
    \begin{enumerate}
        \item Slater’s condition (\cref{def:Slater}) holds for each vehicle’s problem in the lower-level game \cref{eqn:IV-method-bilevel-auxillary-gnep}. \label{method-prop-1}
        \item $\equality(\mathbf{\Breve{\xstate}}^j,\mathbf{\Breve{\control}}^j)$, $\inequality(\mathbf{\Breve{\xstate}}^j,\mathbf{\Breve{\control}}^j)$ in \cref{eqn:IV-method-auxillary-reform-constraint1} and $\equality(\mathbf{\xstate}^j,\mathbf{\control}^j)$ in \cref{eqn:IV-method-auxillary-reform-constraint3} are linear. \label{method-prop-2}
        \item The interval relaxation of binary variables $\mathbf{\binvarz}_u$ results in a single-level convex \ac{qp} problem. \label{method-prop-4}
        \item The lower-level cost function $\cost{j}_l(\mathbf{\binvarz}_u,\Breve{\bvar}_l)$ and the upper-level cost function $\cost{}_u(\Breve{\bvar}_u, \bvar_l)$ are convex and quadratic. \label{method-prop-3}
    \end{enumerate}
\end{proposition}

\begin{proof}
    By construction, the formulation in \cref{eqn:IV-method-bilevel-auxillary-reform} is convex and its feasible set has a non-empty interior since Slater's condition holds for the lower-level game (\cref{method-prop-1,method-prop-2} from \cref{proposition-convergence}). The subproblems in the queue $\mathcal{Q}$ employ continuous interval relaxations of binary variables $\mathbf\binvarz_u$ and thus are convex. This ensures that each solution to the relaxed subproblem obtained at line \ref{alg:relaxed_problem} provides a valid lower bound on the objective \cref{eqn:IV-method-bilevel-auxillary-top-level-obj-reform}. At any node, if a solution is found that improves upon the current best, then the current best solution is updated, which ensures that the sequence of candidate optimal values ($\cost{k*}_u$) is decreasing monotonically. 
    Since there are a finite number of subproblems to explore---$2^{\players \nodes \waypoints}$---\cref{branch-and-bound} terminates in a finite number of steps with a globally optimal solution. 
\end{proof}
\begin{remark}
    \cref{branch-and-bound} can be enhanced through several strategies to improve performance. Some of these strategies include variable selection rules \cite{achterberg2005branching} and node selection heuristics \cite{linderoth1999computational} that affect the search. Additionally, cutting plane generation \cite{costa2009benders} can tighten the relaxation at each node, potentially reducing the number of nodes explored. 
\end{remark}
    

\section{Experimental Results} \label{sec:experiment}
In this section, we showcase the performance of our proposed hierarchical routing framework in a small example, leveraging the reformulation in \cref{eqn:IV-method-bilevel-auxillary-reform}. 
We implement \cref{branch-and-bound} with the PATH ~\cite{dirkse1995path} solver in the Julia programming language.
We consider a formation flight scenario involving two vehicles and three vehicles to demonstrate how the vehicles' interactions influence their individual trajectories in accordance with the high-level routing decisions. 
\subsection{Experiment Setup}

We first formalize vehicle dynamics at the lower level, which are the equality constraints in \cref{eqn:III-prelim-lqgame-constraint1,eqn:III-prelim-lqgame-constraint2}.  
Concretely, we model each vehicle as a double integrator with $\xstate^j_{k,t} = [\xposition_{k,t}^{j,x}, \xposition_{k,t}^{j,y}, \velocity_{k,t}^{j,x}, \velocity_{k,t}^{j,y}]^\top \in \R^4$ encoding the state of the vehicle, and $\control^j_{k,t} = [\acceleration_{k,t}^{j,x}, \acceleration_{k,t}^{j,y}]^\top \in \R^2$ representing the control input at time step $t \in [\horizon].$ 
The state vector $\xstate^j_t$ consists of position and velocity in the horizontal and vertical directions, and the control vector $\control^j_t$ consists of acceleration in the horizontal and vertical directions, respectively. 
We discretize the double-integrator dynamics at resolution $\Delta t$, i.e., 
\begin{equation}
\label{eqn:V-exp-dynamics}
\xstate^j_{k,t+1} = \underbrace{\begin{bmatrix}
    1 & 0 & \Delta t & 0\\
    0 & 1 & 0 & \Delta t\\
    0 & 0 & 1 & 0\\
    0 & 0 & 0 & 1
\end{bmatrix}}_{A} 
\xstate_{k,t}^j + \underbrace{\begin{bmatrix}
    \frac{1}{2}\Delta t^2 & 0\\
    0 & \frac{1}{2}\Delta t^2\\
    \Delta t & 0\\
    0 & \Delta t
\end{bmatrix}}_{B} u^j_{k,t}.
\end{equation}
Note that \cref{eqn:V-exp-dynamics} encodes the constraint \cref{eqn:III-prelim-lqgame-constraint1}.
A similar variant of \cref{eqn:V-exp-dynamics} can be formed for \cref{eqn:III-prelim-lqgame-constraint2}.
Furthermore, we assume the state and control variables to be bounded by $-1 \leq \xstate_{k,t}^j, \control_{k,t}^j \leq 1$ for all vehicles.
 
Next, we define the individual vehicles' objective functions \cref{eqn:III-prelim-lqgame-obj}.
The component of each vehicle's objective which pertains only to their upper- and lower-level variables is encoded by the function $Q(\xstate^j_{k,t},\binvarz_{S^j:T^j,k}^j) \coloneqq \frac{1}{2} \sum_{k=1}^{\waypoints}  \|{\mathbf{p}}^j_{k,1} - \sum_{i=S^j}^{T^j} z_{i,k}^j\hat{\mathbf{p}}_i\|_2^2$.
This cost term penalizes any deviations of the vehicle's state at the start of each trajectory segment $k$, $\mathbf{p}^j_{k,1} = (\xstate_{k,1}^{j,x}, \xstate_{k,1}^{j,y})$, from the \emph{chosen} waypoint, whose location is determined by the product $\binvarz_{i,k}^j \hat{\mathbf{p}}_i$.  

We introduce $I(\mathbf\xstate^j, \mathbf\xstate^{\neg j})$ to account for inter-agent interaction for the lower-level game. 
We model this interaction as $I(\mathbf\xstate^j, \mathbf\xstate^{\neg j}) = \sum_{j' \neq j}^\players \sum_{k=1}^\waypoints \sum_{t=1}^\horizon \| C(\mathbf{\xstate}_{k,t}^j - \mathbf{\xstate}_{k,t}^{j'} + \mathbf{r}^{j j'})\|_2^2$. 
In this section, we consider a one-dimensional variant of $I(\mathbf\xstate^j, \mathbf\xstate^{\neg j})$, i.e., $C = \begin{bmatrix}
    1 & 0 & 0 & 0 
\end{bmatrix}$.
This means that $C\mathbf{\xstate}_{k,t}^j = \xposition_{k,t}^{j,x}$.
In the two-vehicle setting (\cref{fig:comparison-formation-flight-two-vehicles}), we consider $C\mathbf{r}^{12} = r,$ and $C\mathbf{r}^{21} = -r$.
Here, $r>0$ models the situation where Vehicle $1$ wishes to stay on the left of Vehicle $2$ \emph{and} Vehicle $2$ wants to be on the right of Vehicle $1$, {with a separation of $r$ along the horizontal axis.
In the three-vehicle setting (\cref{fig:comparison-formation-flight-three-vehicles}), each of the $\mathbf{r}^{j j'}$ are distinct.

\subsection{Detailed Analysis of Results}
In the following experiment with two vehicles, we consider a scenario with $\nodes = 7$ nodes, and set the number of waypoints to $\waypoints = 5$.
Because their \emph{start} and \emph{terminal} nodes are already given, each vehicle must select $3$ intermediate nodes (out of 7 possibilities) to visit. 
In this example, we set the planning horizon as $\horizon = 7$, i.e., the number of subwaypoints between two waypoints (in a trajectory segment). 
For the three-vehicle case, we set the number of waypoints to $\waypoints = 4$, while keeping all other parameters unchanged. 

\begin{figure}[t]  
    \centering
    \begin{minipage}{\linewidth}
        \centering
        \begin{subfigure}{0.49\linewidth}
            \centering
            \includegraphics[width=\textwidth]{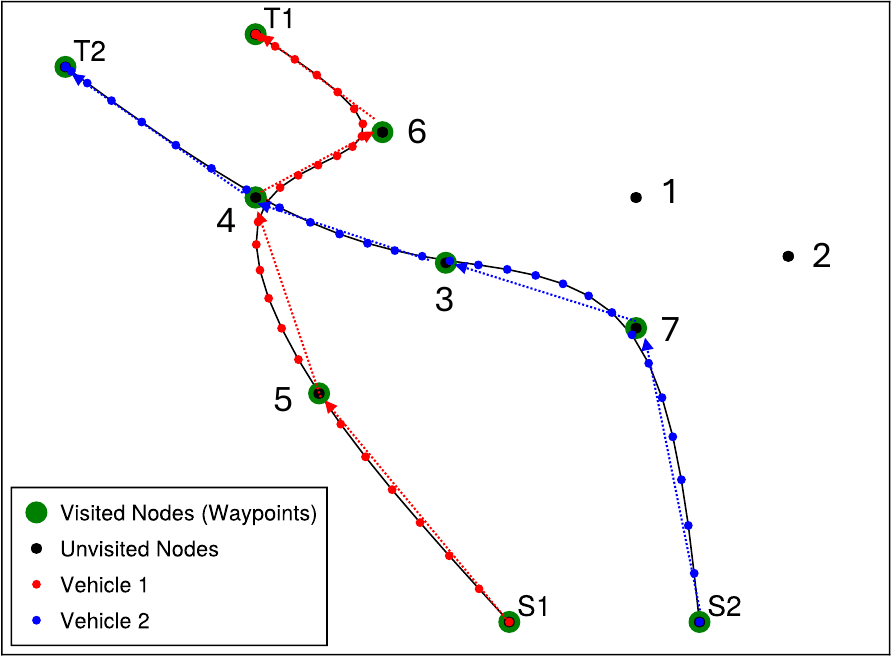}
            \caption{No interaction cost}
            \label{fig:formation_flight0}
        \end{subfigure}
        \hfill
        \begin{subfigure}{0.49\linewidth}
            \centering
            \includegraphics[width=\textwidth]{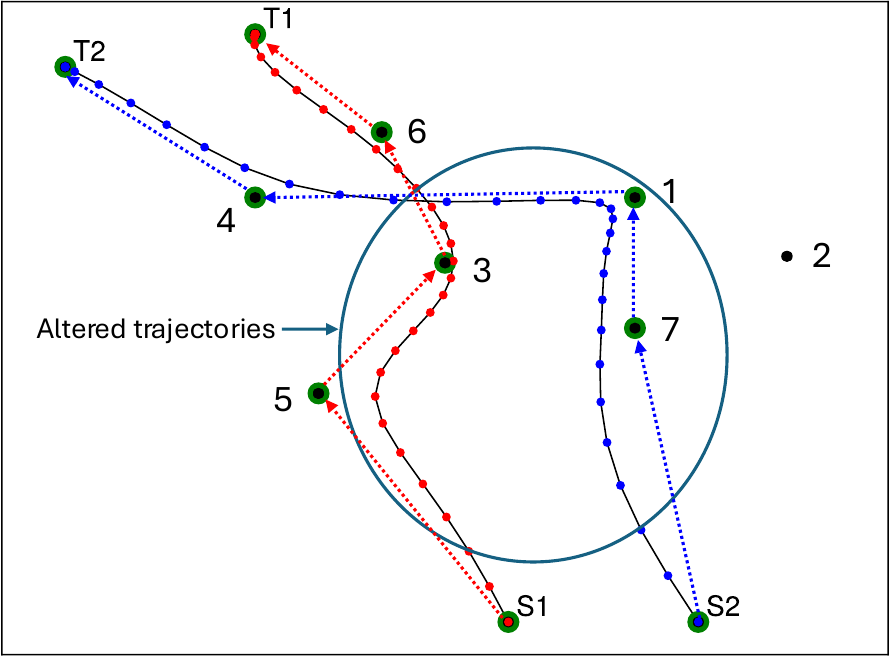}
            \caption{With interaction cost}
            \label{fig:formation_flight1}
        \end{subfigure}
    \end{minipage}

    \caption{Comparison of optimal vehicle routes and trajectories in a two-vehicle formation flight routing game. Dashed straight lines indicate each vehicle’s visited nodes.}
    \label{fig:comparison-formation-flight-two-vehicles}
\end{figure}

\begin{figure}[t]  
    \centering
    \begin{minipage}{\linewidth}
        \centering
        \begin{subfigure}{0.49\linewidth}
            \centering
            \includegraphics[width=\textwidth]{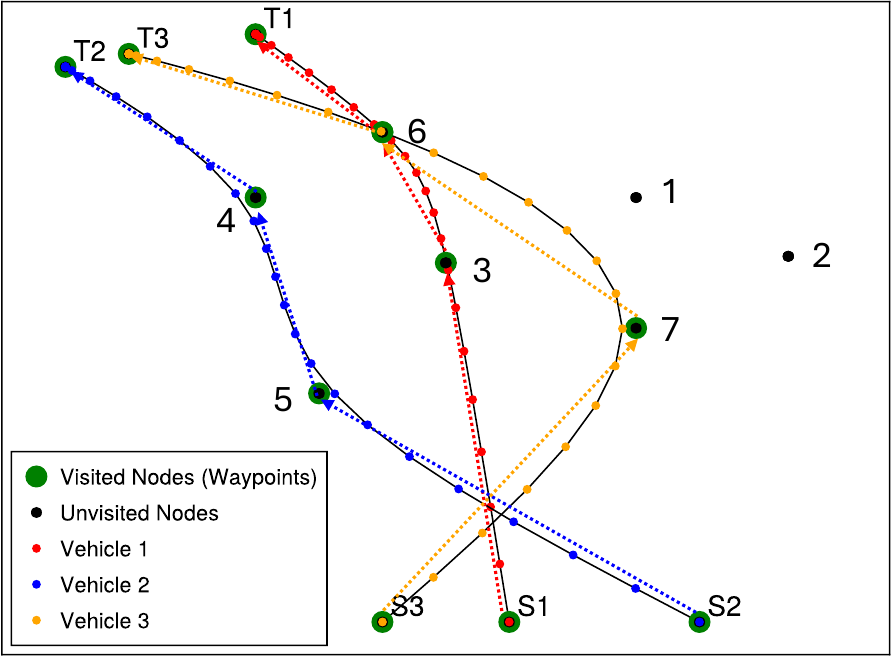}
            \caption{No interaction cost}
            \label{fig:formation_flight2}
        \end{subfigure}
        \hfill
        \begin{subfigure}{0.49\linewidth}
            \centering
            \includegraphics[width=\textwidth]{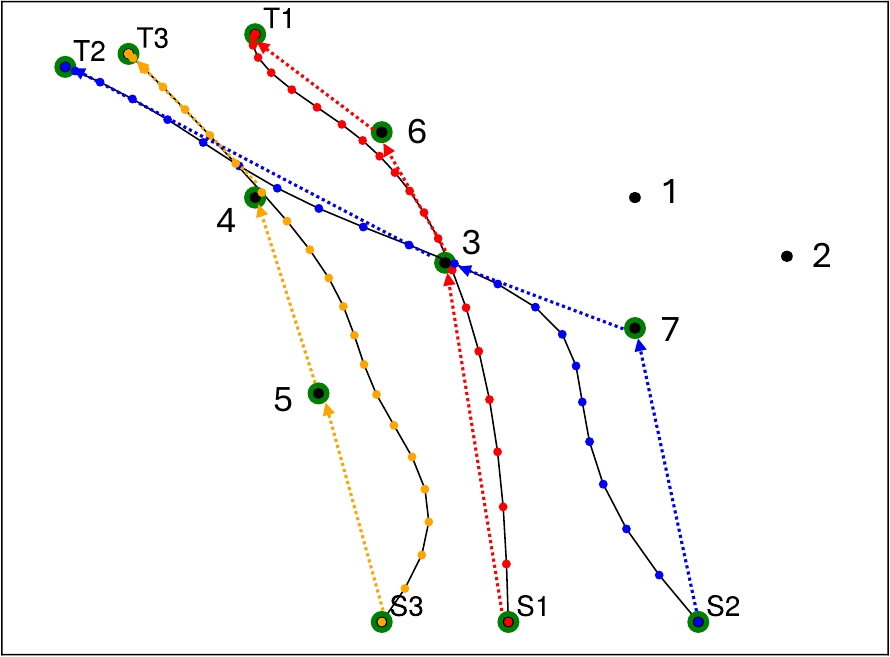}
            \caption{With interaction cost}
            \label{fig:formation_flight3}
        \end{subfigure}
    \end{minipage}

    \caption{Comparison of optimal vehicle routes and trajectories in a three-vehicle formation flight routing game. Dashed straight lines indicate each vehicle’s visited nodes.}
    \label{fig:comparison-formation-flight-three-vehicles}
\end{figure}

\cref{fig:comparison-formation-flight-two-vehicles} and \cref{fig:comparison-formation-flight-three-vehicles} present the impact of interaction costs in \cref{eqn:III-prelim-lqgame} into the lower-level game's objective for each vehicle in two-vehicle and three-vehicle settings, respectively. 
In the absence of $I(\mathbf\xstate^j, \mathbf\xstate^{\neg j})$, each vehicle generates their trajectory to \emph{only} minimize their own cost while visiting required waypoints.
This implies that the game in \cref{eqn:IV-method-bilevel-auxillary} is essentially a set of decoupled optimization problems, for any fixed choice of routing variables $\mathbf\binvarz_u$}.
In constrast, we observe distinct alterations in the trajectories in the presence of the interaction term. 
The circled region in \cref{fig:formation_flight1} highlights the area where trajectories have changed substantially so that the vehicles can maintain formation.
We also observe that the routing selections change as well.
In \cref{fig:formation_flight0}, Vehicle 1 (red) chooses to visit nodes $\mathrm{S}1 \rightarrow 5 \rightarrow 4 \rightarrow 6 \rightarrow \mathrm{T}1$ and Vehicle 2 (blue)'s route is $\mathrm{S}2 \rightarrow 7 \rightarrow 3 \rightarrow 4 \rightarrow \mathrm{T}2$. 
However, accounting for the interaction term, their trajectories change in \cref{fig:formation_flight1}, i.e., Vehicle 1 (red) selects nodes $\mathrm{S}1 \rightarrow 5 \rightarrow 3 \rightarrow 6 \rightarrow \mathrm{T}1$ and Vehicle 2 (blue) visits nodes $\mathrm{S}2 \rightarrow 7 \rightarrow 1 \rightarrow 4 \rightarrow \mathrm{T}2$.
This behavior is a result of the inter-vehicle interaction at the lower level, where each vehicle alters its optimal trajectory in response to the other vehicle's paths.

We observe similar trajectory adjustments in the three-vehicle formation flight scenario.   
Without interaction costs (\cref{fig:formation_flight2}), Vehicle 1 (red) follows the route $\mathrm{S}_1 \rightarrow 3 \rightarrow 6 \rightarrow \mathrm{T}_1$, Vehicle 2 (blue) follows $\mathrm{S}_2 \rightarrow 5 \rightarrow 4 \rightarrow \mathrm{T}_2$, and Vehicle 3 (orange) follows $\mathrm{S}_3 \rightarrow 7 \rightarrow 6 \rightarrow \mathrm{T}_3$.
However, when accounting for inter-agent proximity (in the horizontal direction), the vehicles adjust their trajectories to maintain the designated formation \cref{fig:formation_flight3}.  
Vehicle 1 (red) remains in the middle and makes minimal modifications to its trajectory.  
On the other hand, Vehicles 2 (blue) and 3 (orange) significantly alter their paths to ensure proper spacing from Vehicle 1.  
To this end, Vehicle 2 (blue) visits the nodes $\mathrm{S}_2 \rightarrow 7 \rightarrow 3 \rightarrow \mathrm{T}_2$ and Vehicle 3 (orange) follows the path $\mathrm{S}_3 \rightarrow 5 \rightarrow 4 \rightarrow \mathrm{T}_3$.

We emphasize that these routing changes occur because our formulation explicitly couples the choice of vehicles' waypoints at the high level with their low-level planning process.
The examples above illustrate that, when vehicles' trajectory planning preferences change (e.g., to include a preference for flying in formation), the hierarchical router can \emph{proactively} generate a route which accounts for that change.
Furthermore, due to the convexity of the proposed formulation and by \cref{proposition-convergence}, we are assured that this new route is globally optimal.
\section{Conclusion} \label{sec:conclusion}

In this paper, we present a bilevel game-theoretic framework for multi-agent hierarchical routing that integrates discrete route assignments with continuous trajectory optimization.
We propose a convex reformulation of the resulting mixed-integer program, and demonstrate our results in two- and three-vehicle formation flight scenarios.
Future work will focus on extending the proposed framework to real-world applications, particularly in next-generation air traffic management decision support tools. For example, the US Federal Aviation Administration seeks to adopt a new automation system, \emph{Flow Management Data and Services} (FMDS), which is a potentially multi-hundreds of millions USD investment \cite{faa_fmds}. FMDS seeks to integrate \emph{data streams} such as real- and near-real time aircraft trajectory information with \emph{services} such as air traffic demand-capacity balancing. The modeling and algorithmic results herein provide a foundation for future services that could be run on FMDS, potentially impacting tens of thousands of commercial flights serviced by FMDS.

\bibliographystyle{IEEEtran}
{\footnotesize\bibliography{6_ref.bib}}

\begin{thebibliography}{10}
\providecommand{\url}[1]{#1}
\csname url@samestyle\endcsname
\providecommand{\newblock}{\relax}
\providecommand{\bibinfo}[2]{#2}
\providecommand{\BIBentrySTDinterwordspacing}{\spaceskip=0pt\relax}
\providecommand{\BIBentryALTinterwordstretchfactor}{4}
\providecommand{\BIBentryALTinterwordspacing}{\spaceskip=\fontdimen2\font plus
\BIBentryALTinterwordstretchfactor\fontdimen3\font minus \fontdimen4\font\relax}
\providecommand{\BIBforeignlanguage}[2]{{%
\expandafter\ifx\csname l@#1\endcsname\relax
\typeout{** WARNING: IEEEtran.bst: No hyphenation pattern has been}%
\typeout{** loaded for the language `#1'. Using the pattern for}%
\typeout{** the default language instead.}%
\else
\language=\csname l@#1\endcsname
\fi
#2}}
\providecommand{\BIBdecl}{\relax}
\BIBdecl

\bibitem{zeng2014solving}
B.~Zeng and Y.~An, ``Solving bilevel mixed integer program by reformulations and decomposition,'' \emph{Optimization online}, pp. 1--34, 2014.

\bibitem{kleinert2021survey}
T.~Kleinert, M.~Labb{\'e}, I.~Ljubi{\'c}, and M.~Schmidt, ``A survey on mixed-integer programming techniques in bilevel optimization,'' \emph{EURO Journal on Computational Optimization}, vol.~9, p. 100007, 2021.

\bibitem{richetta1994dynamic}
O.~Richetta and A.~R. Odoni, ``Dynamic solution to the ground-holding problem in air traffic control,'' \emph{Transportation research part A: Policy and practice}, vol.~28, no.~3, pp. 167--185, 1994.

\bibitem{vranas1994multi}
P.~B. Vranas, D.~J. Bertsimas, and A.~R. Odoni, ``The multi-airport ground-holding problem in air traffic control,'' \emph{Operations Research}, vol.~42, no.~2, pp. 249--261, 1994.

\bibitem{bertsimas1998air}
D.~Bertsimas and S.~S. Patterson, ``The air traffic flow management problem with enroute capacities,'' \emph{Operations research}, vol.~46, no.~3, pp. 406--422, 1998.

\bibitem{wu2025managing}
H.~Wu, M.~Z. Li, J.~Henderson, E.~M. Bongo, and L.~A. Weitz, ``Managing congestion in advanced air mobility operations using a bi-level optimization approach,'' in \emph{AIAA SCITECH 2025 Forum}, 2025, p. 0582.

\bibitem{taye2024safe}
A.~G. Taye, R.~Valenti, A.~Rajhans, A.~Mavrommati, P.~J. Mosterman, and P.~Wei, ``Safe and scalable real-time trajectory planning framework for urban air mobility,'' \emph{Journal of Aerospace Information Systems}, vol.~21, no.~8, pp. 641--650, 2024.

\bibitem{zhu2014game}
M.~Zhu, M.~Otte, P.~Chaudhari, and E.~Frazzoli, ``Game theoretic controller synthesis for multi-robot motion planning part i: Trajectory based algorithms,'' in \emph{2014 IEEE International Conference on Robotics and Automation (ICRA)}.\hskip 1em plus 0.5em minus 0.4em\relax IEEE, 2014, pp. 1646--1651.

\bibitem{williams2023distributed}
Z.~Williams, J.~Chen, and N.~Mehr, ``Distributed potential ilqr: Scalable game-theoretic trajectory planning for multi-agent interactions,'' in \emph{2023 IEEE International Conference on Robotics and Automation (ICRA)}.\hskip 1em plus 0.5em minus 0.4em\relax IEEE, 2023, pp. 01--07.

\bibitem{9197129}
D.~Fridovich-Keil, E.~Ratner, L.~Peters, A.~D. Dragan, and C.~J. Tomlin, ``Efficient iterative linear-quadratic approximations for nonlinear multi-player general-sum differential games,'' in \emph{2020 IEEE International Conference on Robotics and Automation (ICRA)}, 2020, pp. 1475--1481.

\bibitem{di2019newton}
B.~Di and A.~Lamperski, ``Newton’s method and differential dynamic programming for unconstrained nonlinear dynamic games,'' in \emph{2019 IEEE 58th conference on decision and control (CDC)}.\hskip 1em plus 0.5em minus 0.4em\relax IEEE, 2019, pp. 4073--4078.

\bibitem{forrestsiopt}
\BIBentryALTinterwordspacing
F.~Laine, D.~Fridovich-Keil, C.-Y. Chiu, and C.~Tomlin, ``The computation of approximate generalized feedback nash equilibria,'' \emph{SIAM Journal on Optimization}, vol.~33, no.~1, pp. 294--318, 2023. [Online]. Available: \url{https://doi.org/10.1137/21M142530X}
\BIBentrySTDinterwordspacing

\bibitem{10021943}
N.~Mehr, M.~Wang, M.~Bhatt, and M.~Schwager, ``Maximum-entropy multi-agent dynamic games: Forward and inverse solutions,'' \emph{IEEE Transactions on Robotics}, vol.~39, no.~3, pp. 1801--1815, 2023.

\bibitem{10160799}
E.~L. Zhu and F.~Borrelli, ``A sequential quadratic programming approach to the solution of open-loop generalized nash equilibria,'' in \emph{2023 IEEE International Conference on Robotics and Automation (ICRA)}, 2023, pp. 3211--3217.

\bibitem{fisac2019hierarchical}
J.~F. Fisac, E.~Bronstein, E.~Stefansson, D.~Sadigh, S.~S. Sastry, and A.~D. Dragan, ``Hierarchical game-theoretic planning for autonomous vehicles,'' in \emph{2019 International conference on robotics and automation (ICRA)}.\hskip 1em plus 0.5em minus 0.4em\relax IEEE, 2019, pp. 9590--9596.

\bibitem{JI2023104109}
\BIBentryALTinterwordspacing
K.~Ji, N.~Li, M.~Orsag, and K.~Han, ``Hierarchical and game-theoretic decision-making for connected and automated vehicles in overtaking scenarios,'' \emph{Transportation Research Part C: Emerging Technologies}, vol. 150, p. 104109, 2023. [Online]. Available: \url{https://www.sciencedirect.com/science/article/pii/S0968090X23000980}
\BIBentrySTDinterwordspacing

\bibitem{marinakis2007new}
Y.~Marinakis, A.~Migdalas, and P.~M. Pardalos, ``A new bilevel formulation for the vehicle routing problem and a solution method using a genetic algorithm,'' \emph{Journal of Global Optimization}, vol.~38, pp. 555--580, 2007.

\bibitem{fan2011bi}
W.~Fan and R.~B. Machemehl, ``Bi-level optimization model for public transportation network redesign problem: Accounting for equity issues,'' \emph{Transportation Research Record}, vol. 2263, no.~1, pp. 151--162, 2011.

\bibitem{bianco2009bilevel}
L.~Bianco, M.~Caramia, and S.~Giordani, ``A bilevel flow model for hazmat transportation network design,'' \emph{Transportation Research Part C: Emerging Technologies}, vol.~17, no.~2, pp. 175--196, 2009.

\bibitem{basar1998}
\BIBentryALTinterwordspacing
T.~Başar and G.~J. Olsder, \emph{Dynamic Noncooperative Game Theory, 2nd Edition}.\hskip 1em plus 0.5em minus 0.4em\relax Society for Industrial and Applied Mathematics, 1998. [Online]. Available: \url{https://epubs.siam.org/doi/abs/10.1137/1.9781611971132}
\BIBentrySTDinterwordspacing

\bibitem{dirkse1995path}
S.~P. Dirkse and M.~C. Ferris, ``The path solver: a nommonotone stabilization scheme for mixed complementarity problems,'' \emph{Optimization methods and software}, vol.~5, no.~2, pp. 123--156, 1995.

\bibitem{dempe2013bilevel}
S.~Dempe and A.~B. Zemkoho, ``The bilevel programming problem: reformulations, constraint qualifications and optimality conditions,'' \emph{Mathematical Programming}, vol. 138, pp. 447--473, 2013.

\bibitem{dempe2015bilevel}
S.~Dempe, V.~Kalashnikov, G.~A. P{\'e}rez-Vald{\'e}s, and N.~Kalashnykova, ``Bilevel programming problems,'' \emph{Energy Systems. Springer, Berlin}, vol.~10, pp. 978--3, 2015.

\bibitem{scholtes2001convergence}
S.~Scholtes, ``Convergence properties of a regularization scheme for mathematical programs with complementarity constraints,'' \emph{SIAM Journal on Optimization}, vol.~11, no.~4, pp. 918--936, 2001.

\bibitem{schwartz2011mathematical}
A.~Schwartz, ``Mathematical programs with complementarity constraints: Theory, methods and applications,'' Ph.D. dissertation, Universit{\"a}t W{\"u}rzburg, 2011.

\bibitem{belotti2013mixed}
P.~Belotti, C.~Kirches, S.~Leyffer, J.~Linderoth, J.~Luedtke, and A.~Mahajan, ``Mixed-integer nonlinear optimization,'' \emph{Acta Numerica}, vol.~22, pp. 1--131, 2013.

\bibitem{kronqvist2019review}
J.~Kronqvist, D.~E. Bernal, A.~Lundell, and I.~E. Grossmann, ``A review and comparison of solvers for convex minlp,'' \emph{Optimization and Engineering}, vol.~20, pp. 397--455, 2019.

\bibitem{achterberg2005branching}
T.~Achterberg, T.~Koch, and A.~Martin, ``Branching rules revisited,'' \emph{Operations Research Letters}, vol.~33, no.~1, pp. 42--54, 2005.

\bibitem{linderoth1999computational}
J.~T. Linderoth and M.~W. Savelsbergh, ``A computational study of search strategies for mixed integer programming,'' \emph{INFORMS Journal on Computing}, vol.~11, no.~2, pp. 173--187, 1999.

\bibitem{costa2009benders}
A.~M. Costa, J.-F. Cordeau, and B.~Gendron, ``Benders, metric and cutset inequalities for multicommodity capacitated network design,'' \emph{Computational Optimization and Applications}, vol.~42, pp. 371--392, 2009.

\bibitem{faa_fmds}
\BIBentryALTinterwordspacing
{Federal Aviation Administration}, ``Flow management data and services (fmds),'' 2025, accessed: 2025-03-16. [Online]. Available: \url{https://www.faa.gov/fmds}
\BIBentrySTDinterwordspacing

\end{thebibliography}


\end{document}